\documentclass[11pt]{article}

\usepackage{amsthm}
\usepackage{graphicx} 
\usepackage{array} 

\usepackage{amsmath, amssymb, amsfonts, verbatim}
\usepackage{hyphenat,epsfig,subfigure,multirow}
\usepackage{nicefrac}
\usepackage{paralist}

\usepackage[usenames,dvipsnames]{xcolor}
\usepackage[ruled, linesnumbered]{algorithm2e}

\usepackage{import}
\usepackage{dsfont}

\usepackage{xargs}
\usepackage[colorinlistoftodos,prependcaption,textsize=small]{todonotes}
\newcommandx{\unsure}[2][1=]{\todo[linecolor=red,backgroundcolor=red!25,bordercolor=red,#1]{#2}}

\DeclareMathOperator{\polylog}{polylog}

\newcommand{\norm}[1]{\left\lVert#1\right\rVert}

\makeatletter
\g@addto@macro\bfseries{\boldmath}
\makeatother

\DeclareFontFamily{U}{mathx}{\hyphenchar\font45}
\DeclareFontShape{U}{mathx}{m}{n}{
      <5> <6> <7> <8> <9> <10>
      <10.95> <12> <14.4> <17.28> <20.74> <24.88>
      mathx10
      }{}
\DeclareSymbolFont{mathx}{U}{mathx}{m}{n}
\DeclareMathSymbol{\bigtimes}{1}{mathx}{"91}

\usepackage{tcolorbox}
\tcbuselibrary{skins,breakable}
\tcbset{enhanced jigsaw}

\usepackage[normalem]{ulem}
\usepackage[compact]{titlesec}

\definecolor{DarkRed}{rgb}{0.5,0.1,0.1}
\definecolor{DarkBlue}{rgb}{0.1,0.1,0.5}

\usepackage{nameref}
\definecolor{ForestGreen}{rgb}{0.1333,0.5451,0.1333}
\definecolor{Red}{rgb}{0.9,0,0}
\usepackage[linktocpage=true,
	pagebackref=true,colorlinks,
	linkcolor=DarkRed,citecolor=ForestGreen,
	bookmarks,bookmarksopen,bookmarksnumbered]
	{hyperref}
\usepackage[noabbrev,nameinlink]{cleveref}
\crefname{property}{property}{Property}
\creflabelformat{property}{(#1)#2#3}
\crefname{equation}{eq}{Eq}
\creflabelformat{equation}{(#1)#2#3}

\usepackage{bm}
\usepackage{url}
\usepackage{xspace}
\usepackage[mathscr]{euscript}

\usepackage{tikz}
\usetikzlibrary{arrows}
\usetikzlibrary{arrows.meta}
\usetikzlibrary{shapes}
\usetikzlibrary{backgrounds}
\usetikzlibrary{positioning}
\usetikzlibrary{decorations.markings}
\usetikzlibrary{patterns}
\usetikzlibrary{calc}
\usetikzlibrary{fit}
\usetikzlibrary{decorations}

\usepackage{mdframed}

\usepackage[noend]{algpseudocode}
\makeatletter
\def\BState{\State\hskip-\ALG@thistlm}
\makeatother

\usepackage[noadjust]{cite}
\usepackage{enumitem}

\usepackage[margin=1in]{geometry}

\newtheorem{theorem}{Theorem}
\newtheorem{lemma}{Lemma}[section]

\newtheorem{corollary}[lemma]{Corollary}
\newtheorem{claim}[lemma]{Claim}

\newtheorem{definition}[lemma]{Definition}

\newtheorem*{claim*}{Claim}
\newtheorem*{proposition*}{Proposition}
\newtheorem*{lemma*}{Lemma}
\newtheorem*{problem*}{Problem}

\crefname{lemma}{Lemma}{Lemmas}
\crefname{claim}{Claim}{Claims}

\newtheorem{mdresult}{Result}
\newenvironment{result}{\begin{mdframed}[backgroundcolor=lightgray!40,topline=false,rightline=false,leftline=false,bottomline=false,innertopmargin=2pt]\begin{mdresult}}{\end{mdresult}\end{mdframed}}

\newtheorem{remark}[lemma]{Remark}
\newtheorem{observation}[lemma]{Observation}

\newtheoremstyle{restate}{}{}{\itshape}{}{\bfseries}{~(restated).}{.5em}{\thmnote{#3}}
\theoremstyle{restate}

\allowdisplaybreaks

\def\congest{\textnormal{\textsf{CONGEST}}\xspace}
\def\local{\textnormal{\textsf{LOCAL}}\xspace}

\renewcommand{\qed}{\nobreak \ifvmode \relax \else
      \ifdim\lastskip<1.5em \hskip-\lastskip
      \hskip1.5em plus0em minus0.5em \fi \nobreak
      \vrule height0.75em width0.5em depth0.25em\fi}

\begin{document}

\title{Improved Bounds for Distributed Load Balancing\footnote{An extended abstract of this paper appears in {DISC 2020}.}}
\author{
  Sepehr Assadi\footnote{(sepehr.assadi@rutgers.edu) Department of Computer Science, Rutgers University.}
  \and
  Aaron Bernstein\footnote{(bernstei@gmail.com) Department of Computer Science, Rutgers University.}
  \and
  Zachary Langley\footnote{(zach.langley@rutgers.edu) Department of Computer Science, Rutgers University.}
}

\date{}

\maketitle

\pagenumbering{roman}

\begin{abstract}
In the load balancing problem, the input is an $n$-vertex bipartite graph $G = (C \cup S, E)$---where the two sides of the bipartite graph are referred to as the clients and the servers---and a positive 
weight for each client $c \in C$. The algorithm must assign each client $c \in C$ to an adjacent server $s \in S$. The load of a server is then the weighted sum of all the clients assigned to it.
The goal is to compute an assignment that minimizes some function of the server loads, typically either the maximum server load (i.e., the $\ell_{\infty}$-norm) or the $\ell_p$-norm of the server loads.
This problem has a variety of applications and has been widely studied under several different names, including: scheduling with restricted assignment, semi-matching, and distributed backup placement.

\medskip

We study load balancing in the distributed setting.  There are two existing results.
Czygrinow et al.\ [DISC 2012] showed a 2-approximate \textsf{LOCAL} algorithm for unweighted clients with round-complexity $O(\Delta^5)$, where $\Delta$ is the maximum degree of the input graph.
Halld{\'o}rsson et al.\ [SPAA 2015] showed an $O(\log{n}/\log\log{n})$-approximate \textsf{CONGEST} algorithm for unweighted clients and $O(\log^2\!{n}/\log\log{n})$-approximation for weighted clients with round-complexity 
$\polylog(n)$. 

\medskip

In this paper, we show the first distributed algorithms to compute an $O(1)$-approximation to the load balancing problem in $\polylog(n)$ rounds:
\smallskip
\begin{itemize}
	\item In the \textsf{CONGEST} model, we give an $O(1)$-approximation algorithm in $\polylog(n)$ rounds for unweighted clients. For weighted clients,  the approximation ratio is $O(\log{n})$.
	
	\smallskip
	
	\item In the less constrained \textsf{LOCAL} model, we give an $O(1)$-approximation algorithm for weighted clients in $\polylog(n)$ rounds.
\end{itemize}

Our approach also has implications for the standard sequential setting in which we obtain the first $O(1)$-approximation for this problem that runs in near-linear time. A $2$-approximation is already known, but it requires solving a linear program and is hence much slower. Finally, we note that all of our results \textit{simultaneously} approximate all $\ell_p$-norms, including the $\ell_{\infty}$-norm.

\end{abstract}

\clearpage

\setcounter{tocdepth}{3}
\tableofcontents

\clearpage

\pagenumbering{arabic}
\setcounter{page}{1}

\section{Introduction}

 \newcommand{\sepehr}[1]{\textcolor{blue}{[\textbf{Sepehr:} #1]}}

In this paper, we study the \emph{load balancing} problem.  The input is a
bipartite graph $G=(C\cup S,E)$, where we refer to the sets $C$ and $S$ the
\emph{clients} and \emph{servers}, respectively.  The goal is to find an
assignment of clients to servers such that no server is assigned too many
clients.  To be more precise, we define the load of a server in an assignment
to be the number of clients assigned to it, and we are interested in finding an
assignment that minimizes the maximum load of any server (or minimizes the
$\ell_p$-norm of the server loads---we will discuss this objective more later in the introduction).

The load balancing problem has a rich history in the scheduling literature as
the \textit{job scheduling with restricted assignment
problem}~\cite{H73,BCS74,LL04,HLLT06}, in the distributed computing literature
as the \textit{backup placement problem}~\cite{HKPR18,OBL18,BO20}, in sensor
networks~\cite{SSK06,MT08} and peer-to-peer systems~\cite{KSTZ04,STZ04a,STZ04b}
as the \textit{load-balanced data gathering tree construction problem}, and more
generally as a relaxation of the bipartite matching problem known as the
\textit{{semi-matching} problem}~\cite{HLLT06,CHSW12,FLN14,KR16}.  We refer the
reader to~\cite{HLLT06,FLN14,HKPR18} for more background.

Our primary focus in this paper is on the load balancing problem in a
distributed setting, where clients and servers correspond to separate nodes in
a network.  Communication through the network happens in synchronous rounds,
where in each round, every node can send $O(\log{n})$ bits to its neighbors
over any of its incident edges (formally, we work in the \congest{} model---see
\Cref{sec:prelim} for more details).  In the distributed setting, the
load balancing problem generalizes the distributed backup placement problem
with replication factor one (introduced in~\cite{HKPR18}), where the nodes
(corresponding to clients) in a distributed network may have memory faults and
therefore wish to store backup copies of their data at neighboring nodes
(corresponding to servers). Since backup-nodes may incur faults as well, the
number of nodes that select the same backup-node should be minimized.  See
\Cref{app:dbp} for the exact formulation of the distributed backup
placement problem and for how some of our results extend to the more general
version of the problem with arbitrary replication factor. 

A simple distributed algorithm for the load balancing problem in which the
nodes myopically reassign themselves to a server with smaller load eventually
converges to an $O(\frac{\log{n}}{\log\log{n}})$-approximation, where $n$ is
the number of nodes in the network~\cite{GLMM06,KTM12}, but as was shown in
Halld{\'o}rsson~et~al.~\cite{HKPR18}, the algorithm requires $\Omega(\sqrt{n})$
rounds.  The same paper~\cite{HKPR18} shows a way of circumventing this costly
process and gives a distributed algorithm that achieves the same
$O(\frac{\log{n}}{\log\log{n}})$-approximation in only $\polylog(n)$ rounds.
On the other hand, Czygrinow~et~al.~\cite{CHSW12} show a distributed
$O(1)$-approximation (precisely, a $2$-approximation) that requires
$O(\Delta^5)$ rounds, where $\Delta$ is the maximum degree of a node in the
network.  This algorithm is highly efficient for low-degree networks but is
again too expensive for high-degree graphs.  

This state-of-affairs is the starting point of our work: \emph{Can we obtain
the best of both worlds, namely, an $O(1)$-approximation algorithm in
$\polylog(n)$ rounds?} 

\subparagraph{Our first contribution.} Our first main contribution in this
paper is an affirmative answer to this question. 

\begin{result}[Formalized in \Cref{thm:hkpr}]\label{res:1}
  We give an $O(1)$-approximate randomized distributed algorithm for load
  balancing in the \congest{} model that runs in $O(\log^5\!n)$ rounds. 
\end{result}
	
At the core of our algorithm is a new structural lemma for the load balancing
problem.  Informally speaking, we show that eliminating all ``short augmenting
paths'' of length $O(\log{n})$ is sufficient to assign \textit{all} clients to
servers with load a constant factor as much as the optimum
(\Cref{lem:expansion}).  In conjunction with ideas
from~\cite{HKPR18}, this effectively reduces the load balancing problem to that
of finding a matching with no short augmenting paths, which can be solved using
the by-now standard algorithm of Lotker~et~al.~\cite{LPP15}.  

\subparagraph{Our second contribution.} Next, we consider the \emph{weighted}
load balancing problem in which every client comes with a weight.  The load of
a server is then the total weight of the clients assigned to it.  The goal, as
before, is to minimize the maximum load of any server.
Halld{\'o}rsson~et~al.~\cite{HKPR18} also studied the weighted problem and gave
an $O(\frac{\log^{2}{\!n}}{\log\log{n}})$-approximation in $\polylog(n)$ rounds
using a simple reduction to the unweighted case. 

Using the same weighted-to-unweighted reduction, our algorithm in
\Cref{res:1} also implies an $O(\log{n})$-approximation for the weighted
load balancing problem in $\polylog(n)$ rounds of the \congest{} model.  Our
main technical contribution in this paper is a new algorithm for this problem
that achieves an $O(1)$-approximation in the less constrained \local{} model,
in which communication over edges in each round is unbounded. 

\begin{result}[Formalized in \Cref{thm:pnorm-weighted}]\label{res:2}
  We give an $O(1)$-approximate randomized distributed algorithm for weighted
  load balancing in $O(\log^3\!{n})$ rounds of the \local{} model. 
\end{result}

Our \local{} algorithm consists of two main components: a distributed algorithm
for (approximately) solving a relaxed version of the problem where each client
$c$ with weight $w(c)$ should be assigned to $w(c)$ adjacent servers with
multiplicity---a \emph{split assignment}---and a novel distributed rounding
procedure.  Using our structural result in \Cref{lem:expansion}, we can
find a split assignment by approximately solving (or rather, eliminating short
augmenting paths in) a generalized $b$-matching problem with \emph{edge
capacities}. We are not aware of any efficient algorithm for this problem in
the \congest{} model, but we can show that a simple extension of the work
of~\cite{LPP15} can solve this problem in $\polylog(n)$ rounds in the \local{}
model.  The rounding step is also based on a new application of our
\Cref{lem:expansion} that allows us to circumvent the typical use of
``cycle canceling'' procedures for rounding fractional matching LP solutions
into integral ones, which do not translate to efficient distributed algorithms.

\medskip
We now turn to two important extensions of \Cref{res:1} and \Cref{res:2}.  The
first is the more general problem of all-norm load balancing, and the second is
a fast sequential algorithm.

\subparagraph{Approximating all norms.} Recall that our goal in the load
balancing problem has been to minimize the maximum load of any server. Assuming
we denote the loads of servers under some assignment $A$ by a vector $L_A :=
[L_A(s_1),L_A(s_2),\ldots,L_A(s_n)]$ for all $s_i \in S$, minimizing the
maximum server load is equivalent to minimizing $\norm{L_A}_{\infty}$, i.e.,
the $\ell_\infty$-norm of $L_A$.  Depending on the application, however,
miniziming this norm may not be the most natural notion of a ``balanced''
assignment; if some server requires vastly more load than the other servers, an
$\ell_\infty$-norm-minimizing assignment may put needlessly large load on those
other servers.

As a result, it is natural to consider minimizing some other $\ell_p$-norm of
$L_A$ for some $p \geq 1$.  This is done, for instance,
in~\cite{HLLT06,FLN14,KR16}, which considered $\ell_2$-norms. An even more
general objective is the \emph{all-norm} problem, studied
in~\cite{AERW04,BKPPS17,CS19,HLLT06}, where the goal is to
\emph{simultaneously} optimize with respect to every $\ell_p$-norm.  These
results compute an assignment which is an $O(1)$-approximation (or even
optimal) simultaneously with respect to all $\ell_p$-norms, including
$p=\infty$ (\textit{a priori}, even the existence of such an assignment is not
clear).

All of our results extend to the all-norm problem without any increase in
approximation factor or round-complexity.  In particular, in the \congest{}
model, we give randomized distributed $O(1)$- and $O(\log{n})$-approximation
algorithms for all-norm load balancing in $\polylog(n)$ rounds, in the
unweighted and weighted variant of the problem, respectively
({\Cref{thm:unweighted-pnorm}).  In the \local{} model, the
approximation ratio for the weighted problem can be reduced to $O(1)$ as well
(\Cref{thm:pnorm-weighted}). 

\subparagraph{Faster sequential algorithms.} Finally, we show that our new
approach to weighted load balancing can also be used to design a near-linear
time algorithm for this problem in the \emph{sequential} setting.  We give a
deterministic $O(m\log^{3}\!{(n)})$ time algorithm for the $O(1)$-approximate
all-norm load balancing problem in the sequential setting
(\Cref{thm:sequential}). 

Previously, a deterministic $O(m\sqrt{n}\log{n})$ time for the exact problem in
case of unweighted graphs was given in~\cite{FLN14}.  The weighted variant of
the problem is NP-hard~\cite{AERW04}; $2$-approximate algorithms were shown
in~\cite{AERW04} and~\cite{CS19}, but they are based on solving, respectively,
the linear and convex programming relaxations of the problem exactly using
the ellipsoid algorithm, and thus are much slower than the algorithm we
present.

\section{Preliminaries}\label{sec:prelim}

\subparagraph{Notation.} For any function $f : A \to \mathbb{N}$ and $B
\subseteq A$, we use the notation $f(B) = \sum_{b \in B} f(b)$ to sum $f$ over
all elements in $B$.  For any integer $t \geq 1$, we denote $[t] :=
\{1,\ldots,t\}$. 

Throughout, we assume $G = (C \cup S, E)$ is a bipartite graph.  We refer to
$C$ and $S$ as the \emph{clients} and the \emph{servers}, respectively.  We let
$uv$ denote the edge between vertices $u$ and $v$ and let $\delta(v)$ denote
the set of edges incident to the vertex $v$.  We use $n$ as number of vertices in
$G$ and $m$ as the number of edges in $G$.

\subparagraph{Load balancing.}
In the {load balancing problem}, the input is a bipartite graph $G = (C \cup S,
E)$ together with a {client weight function} $w : C \to [W]$.  The output is an
{assignment} $A : C \to S$ mapping every client to one of its adjacent servers.
The \emph{load} $L_A(s)$ of a server $s \in S$ under assignment $A$ is the sum
of the weights of the clients assigned to it: $L_A(s) = w(A^{-1}(s))$.  The
\emph{maximum load} of an assignment $A$ is the maximum load of any server under
$A$.  We refer to the problem of computing an assignment of minimum load as the
\emph{(weighted) min-max load balancing problem}.

As mentioned in the introduction, the min-max objective can be generalized by
considering any $\ell_p$-norm of $L_A$, defined as $\norm{L_A}_p =
{\left(\sum_{s \in S} {(L_A(s))}^p\right)}^{1/p}$.  For brevity, we also use
the notation $\norm{A}_p := \norm{L_A}_p$.  In the language of norms, the
min-max objective corresponds to minimizing the load vector's
$\ell_\infty$-norm.  When the goal is to find an assignment $A$ that
simultaneously minimizes $\norm{A}_p$ for all $p \ge 1$, including $p=\infty$,
the problem is called the \emph{(weighted) all-norm load balancing problem}.
Prior results in~\cite{AERW04,BKPPS17,CS19,HLLT06} show the \textit{existence}
of an assignment that can (approximately) minimize all these norms
simultaneously.  In particular, we use the following result due to
Alon~et~al.~\cite{AAWY97} in our proofs (see also~\cite{HLLT06, BKPPS17}). 

\begin{lemma}[\cite{AAWY97}]\label{lem:optimal-exists}
  Given any instance of the unweighted load balancing problem, there exists an
  assignment $A^*$ that simultaneously minimizes $\norm{A^*}_p$ for all $p \ge
  1$, including $p = \infty$.
\end{lemma}

\subparagraph{$b$-matchings.}
In addition to assignments, we will also work with \emph{$b$-matchings}.  For a
vertex capacity function $b : V \to \mathbb{Z}^+$, a $b$-matching $x : E
\to\mathbb{Z}^+$ gives to each edge a multiplicity such that $x(\delta(v)) \le
b(v)$ for every vertex $v$.

Since we will focus solely on the case when $G$ is bipartite and $V = C \cup
S$, it will be convenient to split $b$ into two separate capacity functions,
one for the clients and one for the servers.  We use $\kappa : C \to
\mathbb{Z}^+$ to denote the \emph{client capacities} and $\tau : S \to
\mathbb{Z}^+$ to denote the \emph{server capacities}.  A \emph{$(\kappa,
\tau)$-matching} is then a function $x : E \to \mathbb{Z}^+$ assigning
multiplicities to edges such that
\begin{equation}
	\label{eq:bmatching}
	\sum_{s \in N(c)} x(cs) \le \kappa(c)
\end{equation}
for every client $c$ and
\begin{equation*}
	\sum_{c \in N(s)} x(cs) \le \tau(s)
\end{equation*}
for every server $s$.
A $(\kappa, \tau)$-matching is \emph{client-perfect} if \eqref{eq:bmatching}
holds with equality for all $c \in C$.  We say that a server $s$ (resp.\ client
$c$) is \emph{$x$-saturated} if $x(\delta(s)) = \tau(s)$ (resp.\ $x(\delta(c))
= \kappa(c)$).  If a vertex is not $x$-saturated, then it is
\emph{$x$-unsaturated}.  An \emph{$x$-augmenting path} is a path $v_1, \dots,
v_{2k}$ such that $v_1$ and $v_{2k}$ are $x$-unsaturated and $x(v_{2i}
v_{2i+1}) > 0$ for all $i \in [k - 1]$.

We will make repeated use of the following simple remark.
\begin{remark}
  When all client weights are one (the unweighted case), a client-perfect $(1,
  \tau)$-matching induces an assignment of maximum load at most $\max_{s \in S}
  \tau(s)$, and vice versa.
\end{remark}
Note that the remark does not generalize to weighted clients; under a $(w,
\tau)$-matching, a client may be split across multiple servers, which does not
correspond to a proper assignment.

\subparagraph{The \local{} and \congest{} models.}
In both the \local{} and the \congest{} models of distributed computation, each
vertex of the input graph hosts a processor that initially only knows its
neighbors and its weight.  Following a standard assumption, we assume that all
vertices know $n$ and the maximum weight $W$.  Computation proceeds in
synchronous rounds; in each round, vertices may send messages to their
neighbors and then receive messages from their neighbors in lockstep.  Local
computation is free---the performance measure of interest is the \emph{round
complexity}, the number of rounds the algorithm takes to complete.

The \local{} and \congest{} models differ in that in the \local{} model,
vertices can send and receive arbitrarily large messages, while in the
\congest{} model, the communication between adjacent vertices in each round is
capped at $O(\log{n})$.

\section{A Structural Lemma}

A crucial component of our results is a structural observation about
approximate $(\kappa, \tau)$-matchings in the context of the load balancing
problem, which is inspired by results from online load
balancing~\cite{GuptaKS14,BernsteinHR19}: if a graph contains \textit{some}
client-perfect $(\kappa, \tau)$-matching, then \textit{every} $(\kappa,
2\tau)$-matching is either client-perfect or can be augmented via an augmenting
path of logarithmic length.  Formally, and more generally, we have the
following lemma.

\begin{lemma}\label{lem:expansion}
  Suppose $G$ contains a client-perfect $(\kappa, \tau)$-matching and let $x$
  be a $(\kappa, \alpha \tau)$-matching for some $\alpha > 1$. If a client $c$
  is $x$-unsaturated, then there is an $x$-augmenting path starting from $c$ of
  length at most $2\lceil \log_{\alpha}\!\tau(S)\rceil + 1$.
\end{lemma}

\begin{proof}
Suppose $G$ contains a client-perfect $(\kappa, \tau)$-matching $x^*$.  To
simplify the discussion, we define a directed multigraph $D$ on $V(G)$ whose
arcs are oriented edges in the support of $x$ and $x^*$ as follows.  For every
$cs \in E(G)$ with $c \in C$ and $s \in S$, $D$ has $x^*(cs)$ copies of the arc
$(c, s)$, $x(cs)$ copies of the arc $(s, c)$, and no other arcs.  Notice that
every directed path in $D$ starting from $c$ and ending at an $x$-unsaturated
server corresponds to an $x$-augmenting path starting from $c$ in $G$.

Let $k \in \mathbb{N}$ and define $U_k$ to be the set of vertices
reachable via a walk of length $k$ from $c$ in $D$.  Call $U_k$ \emph{full} if
$u$ is $x$-saturated for all $u \in U_k$.

The lemma follows from two simple claims:
\begin{enumerate}
\item if $U_{2k+1}$ is not full, then $G$ contains an $x$-augmenting path from $c$ of
  length at most $2k+1$; and
\item if $U_{2k+1}$ is full, then $\tau(U_{2k+3}) \ge \alpha \tau(U_{2k+1})$.
\end{enumerate}

The first claim follows from the fact that a directed walk contains a directed
path with the same endpoints and from the correspondence noted earlier between
directed paths from $c$ ending at an unsaturated server in $D$ and
$x$-augmenting paths from $c$ in $G$.

We proceed to the second claim.  If $s \in U_{2k+1}$ and $U_{2k+1}$ is full,
then $s$ is $x$-saturated and the out-degree of $s$ is $\alpha\tau(s)$.  Thus,
the total out-degree of $U_{2k+1}$---and also the total in-degree of
$U_{2k+2}$---is $\alpha\tau(U_{2k+1})$.  Now we use the fact that the
out-degree of a client $c \in U_{2k+2}$ is at least as large as its in-degree.
This follows simply from the fact that the in-degree must be at most
$\kappa(c)$, and since $x^*$ is client-perfect, the out-degree is exactly
$\kappa(c)$.  Following the arcs once more, the total in-degree of $U_{2k+3}$
is at least $\alpha\tau(U_{2k+1})$.  Finally, since the in-degree of $U_{2k+3}$
is also point-wise less than $\tau$, we have $\alpha\tau(U_{2k+1}) \le
\tau(U_{2k+3})$.

Now we show how the two claims together imply the lemma.  If $U_{2i+1}$ is not
full for some $i \le \lceil \log_\alpha(\tau(S))\rceil$, we are done by the
first claim.  Otherwise, the sums of capacities grow exponentially starting
with $\tau(U_1) \ge 1$.  Inductively, for $k = \lceil \log_\alpha \tau(S) \rceil + 1$,
we have $\tau(U_{2k+1}) \ge \alpha^k > \tau(S)$, a contradiction.  Thus, not
all $\{U_{2i+1}\}$ are full for $i \le \lceil \log_\alpha(\tau(S))\rceil$ and
so we again apply the first claim.
\end{proof}

\section{Unweighted Load Balancing}

\label{sec:unweighted}
Assuming an algorithm to eliminate augmenting paths up to a certain length
efficiently, the structural lemma from the previous section almost immediately
implies an algorithm for the unweighted load balancing problem.  To eliminate
short augmenting paths we use the following lemma, which is implied by Lemma 24
in~\cite{HKPR18}.

\begin{lemma}[\cite{HKPR18}]\label{lem:hksr}
There exists an $O(k^3\log{n})$-round randomized algorithm in the \congest{}
model that, given a graph $G = (C \cup S, E)$, a positive integer $k$, and
server capacity function $\tau$, computes with high probability a $(1,
\tau)$-matching with no augmenting paths of length $k$ or less.
\end{lemma}

The proof of \Cref{lem:hksr} combines two existing results. The algorithm
of Lotker~et~al.~\cite{LPP15} computes a (1,1)-matching with no augmenting
paths of length $\leq k$ in $O(k^3\log{n})$ rounds.
Halld{\'o}rsson~et~al.~\cite{HKPR18} then show a black-box extension from
(1,1)-matching to (1, $\tau$)-matching which does not increase the
round-complexity; see~\cite{HKPR18} for more details.

\begin{remark}
Both our algorithm and the algorithm of~\cite{HKPR18} use the above lemma as a
starting point, but the algorithm of~\cite{HKPR18} only removes short
augmenting paths to ensure that the $(1,\tau)$-matching is approximately
optimal. Since a near-optimal matching is still not an assignment (as it is not
client-perfect), they then use a different set of tools to convert an
approximate $(1,\tau)$-matching to an $O(\log{n}/\log\log{n})$-approximate
assignment.

Our analysis, by contrast, directly exploits the non-existence of short
augmenting paths via \Cref{lem:expansion}. We thus avoid the additional
conversion of~\cite{HKPR18}, which leads to a better approximation ratio, as
well as a simpler algorithm. 
\end{remark}

\subparagraph{Approximating the $\ell_\infty$-norm (the min-max load balancing problem).}

Let $B^*$ be the optimum $\ell_\infty$-norm.  We will first describe an
algorithm that assumes as input some $B \ge B^*$.  The algorithm begins by
using \Cref{lem:hksr} to compute a $(1, 2B)$-matching $x$ with no
augmenting paths of length $4\lceil\log_2{n}\rceil+1$.  The sum of the server
capacities is at most $nB$, and since clients have unitary weight, we can
assume $B \le n$.  Therefore, by \Cref{lem:expansion}, since there are no
$x$-augmenting paths of length $2\lceil \log_2(nB)\rceil+1 \le 4 \lceil
\log_2{n}\rceil + 1$, we know that $x$ is necessarily client-perfect.  A client
$c$ can now assign itself to the vertex it is matched to under $x$.

To remove the assumption that we are given a $B \ge B^*$, we run the algorithm
above $\log{n}$ times with $B = 1, 2, 4, \dots, n$.  For every run where $B \ge
B^*$, the algorithm will successfully assign every client. Note, however, that
in the distributed setting, there is no efficient way for the clients to
determine the smallest $B$ for which the algorithm successfully matched every
client.  Instead, each client $c$ \textit{locally} assigns itself according to
the run with smallest $B$ that succeeded---i.e., according to the first run in
which $c$ was matched. We show that the resulting assignment has maximum load
at most $8B^*$. See \Cref{alg:congest_unweighted} for a concise
treatment.
 
\begin{figure}[h]
\begin{algorithm}[H]
  \caption{\label{alg:congest_unweighted} Approximate unweighted load balancing in the \congest{} model.}
  \For{$B \in \{1, 2, 4, \dots, n\}$}{
 	compute a $(1, 2B)$-matching $x_B$ with no augmenting paths of length $4\lceil\log{n}\rceil + 1$\label{line:unweighted-augment}
  }
  each client $c$ locally finds the minimum $B$ such that $c$ is matched in $x_B$ and assigns itself to the server it is matched to in $x_B$\label{line:unweighted-assign}\\
\end{algorithm}
\end{figure}

\begin{theorem}\label{thm:hkpr}
  In the \congest{} model, there is a randomized algorithm
  (\Cref{alg:congest_unweighted}) that with high probability computes an
  $O(1)$-approximation to the min-max load balancing problem in
  $O(\log^5\!{n})$ rounds.
\end{theorem}

\begin{proof}
  First observe that since augmenting paths with respect to a $(1, n)$-matching
  have length at most 1,  all clients are assigned in $x_n$.  Therefore the
  algorithm always outputs an assignment of all clients.

  Let $B^*$ be the load of an optimal assignment and let $B$ be the unique
  power of two such that $B^* \le B < 2B^*$.  The $(1, 2B)$-matching $x_B$
  computed in the main loop of \Cref{alg:congest_unweighted} will be
  client-perfect by \Cref{lem:expansion}, and so no client will assign
  itself according to $x_{B'}$ for $B' > B$.  A single server $s$ is only
  assigned at most $2i$ clients from $x_i$, and since $x_1, x_2, x_4, \dots,
  x_B$ are the only assignments contributing to the load of $s$, the total load
  of $s$ is at most $2 + 4 + 8 + \cdots + 2B < 4B \leq 8B^*$.  Thus every
  server has load at most $8B^*$.

  Finally, each $x_B$ is computed in $O(\log^4\!{n})$ rounds with high
  probability by \Cref{lem:hksr}.  We compute them sequentially, resulting in
  total round complexity of $O(\log^5\!{n})$.
\end{proof}

\begin{remark}\label{rem:alg1-local}
In the \local{} model, the round complexity of
\Cref{alg:congest_unweighted} is $O(\log^3\!{n})$.  One $\log$ factor is shaved
off of \Cref{lem:hksr} because the algorithm of~\cite{LPP15} for finding a
(1,1)-matching is faster in the \local{} model. The second $\log$ factor is
shaved off by running the for-loop of \Cref{alg:congest_unweighted} in parallel
for every $B$.
\end{remark}

\subparagraph{Approximating all $\boldsymbol{\ell_p}$-norms simultaneously (the all-norm load balancing problem).}
The assignment produced by \Cref{alg:congest_unweighted} in fact does
more than approximate the optimal $\ell_\infty$-norm; it also simultaneously
approximates every $\ell_p$-norm for $p \geq 1$, as we will now show.

Recall that by \Cref{lem:optimal-exists}, when clients are unweighted,
there is an {all-norm optimal assignment} that simultaneously minimizes the
$\ell_p$-norm for all $p \ge 1$, including $p=\infty$. 
Let $A^*$ be the lexicographically smallest all-norm optimal solution (the
choice is arbitrary; we just need a canonical optimal solution). We will need
the following key definition.
\begin{definition}
  The \emph{level} $\ell(s)$ of a server $s$ is defined as the load of $s$ in
  $A^*$, i.e., $\ell(s) := L_{A^*}(s)$ for all $s \in S$.  We define the
  \emph{level} $\ell(c)$ of a client $c$ as $\ell(c) := L_{A^*}(A^*(c))$, the load
  of the server to which $c$ is assigned in $A^*$.
\end{definition}

To take advantage of some additional structure, we will also need the following
definition. 

\begin{definition}
  A \emph{cost-reducing path} with respect to an assignment $A$ is a path $v_1
  v_2 \cdots v_{2k+1}$ such that
  \begin{enumerate}
    \item $v_1, v_k \in S$;
    \item $A(v_{2i}) = v_{2i-1}$ for all $i \in [k]$; and
    \item $L_A(v_k) \le L_A(v_1) - 2$.
  \end{enumerate}
\end{definition}
As the name implies, the existence of a cost-reducing path with respect to an
assignment $A$ implies that $A$ is not all-norm optimal, as re-assigning
$v_{2i}$ to $v_{2i+1}$ for all $i \in [k]$ produces an assignment with smaller,
e.g., $\ell_2$-norm.

\begin{lemma}[\cite{AAWY97, HLLT06, BKPPS17}]\label{lem:cost-reducing}
  There is no cost-reducing path with respect to the assignment $A^*$. \qed
\end{lemma}

We obtain the following immediate corollary.

\begin{corollary}\label{cor:nodown}
  Let $c$ be a client and let $s$ be a server. If $\ell(s) \le \ell(c) - 2$,
  then $c$ and $s$ are not adjacent.
\end{corollary}
\begin{proof}
  If there did exist an edge $cs$, then $\{A^*(c), c, s\}$ would form a cost-reducing
  path with respect to $A^*$.
\end{proof}

\begin{lemma}\label{lem:sat-opt}
  Let $B$ be a positive integer and suppose $x$ is a $(1, 2B)$-matching
  such that $G$ contains no $x$-augmenting paths of length $4\lceil
  \log{n}\rceil + 1$. If $c \in C$ and $\ell(c) \le B - 1$, then $c$ is saturated
  by $x$.
\end{lemma}

\begin{proof}
  We consider the subgraph $G'$ of $G$ defined as follows. The vertices of $G'$ are
  all clients and servers whose level is at most $B$. That is, \[
    V(G') = \{ u \in V : \ell(v) \le B\}
  \]
  The edges of $G'$ consist only of the edges in the support of $x$ and $A^*$: \[
    E(G') = \{ cs : \text{$x(cs) = 1$ or $A^*(c) = s$}\}.
  \]
  Let $n'$ be the number of vertices in $G'$ and let $x'$ be the restriction of
  $x$ to $G'$.

  Now let $c$ be a client of level at most $B - 1$ and suppose that $c$ is not
  $x$-saturated; as $x'$ is the restriction of $x$, clearly $c$ is not
  $x'$-saturated either. Since $G'$ contains a $(1, B)$-matching (namely, the
  restriction of $A^*$) and since $x'$ is a $(1, 2B)$-matching, by
  \Cref{lem:expansion} there exists an $x'$-augmenting path $P$ from $c$
  of length at most $4 \lceil \log{n'} \rceil + 1$ in $G'$. The graph $G$
  contains no such augmenting paths by assumption, and so the path $P$ ends at
  an $x$-saturated server $s$.

  Because $s$ is $x$-saturated but $x'$-unsaturated, there must be a client $d$
  of level at least $B + 1$ adjacent to $s$. However, notice that the path $P$
  then gives a cost-reducing path with respect to $A^*$. In particular, the
  path formed by starting with the vertices $A^*(d), d, s$ and continuing along
  $P$ to $A^*(c)$ gives a cost-reducing path with respect to $A^*$, 
  contradicting \Cref{lem:cost-reducing}.  Thus $c$ is $x$-saturated.
\end{proof}

\begin{theorem}\label{thm:unweighted-pnorm}
  In the \congest{} model, there is a randomized algorithm
  (\Cref{alg:congest_unweighted}) that with high probability computes
  an $O(1)$-approximation to the all-norm load balancing problem in $O(\log^5\!{n})$
  rounds.
\end{theorem}

\begin{proof}
  We have already shown that \Cref{alg:congest_unweighted} produces an
  assignment (\Cref{thm:hkpr}). It remains to show that the $\ell_p$-norm of
  this assignment is good.

  Fix a server $s$. \Cref{cor:nodown} implies that the load of $s$ is
  only determined by clients whose level is at most $\ell(s) + 1$ as other
  clients are not adjacent to $s$. Let $\hat{B}$ be the unique power of two such
  that \begin{equation}\label{eq:ptwo}
    \ell(s) + 2 \le \hat{B} < 2(\ell(s) + 2).
  \end{equation}
  Notice that by \Cref{lem:sat-opt}, every client of level at most $\ell(s) +
  1$ will be saturated in $x_{\hat{B}}$, and so every client adjacent to $s$
  will choose its server in $x_B$ for some $B \le \hat{B}$. Thus, we bound the
  load of $s$ as $L_A(s) \le 2 + 4 + 8 + \cdots + 2\hat{B} \le 24 \ell(s)$,
  where we have used that $\hat{B} < 6\ell(s)$ by~\eqref{eq:ptwo}.  Hence,
  \[
    \norm{A}_p = {\left(\sum_{s \in S} {(L_A(s))}^{p}\right)}^{1/p} \le
  {\left(\sum_{s \in S} {(24\ell(s))}^p\right)}^{1/p} = 24\norm{A^*}_p.\qedhere
\]
\end{proof}

\section{Weighted Load Balancing}\label{sec:weighted}

\newcommand{\afrac}{y_f}
\newcommand{\afracopt}{\afrac^*}

\SetKwComment{Comment}{$\triangleright$\ }{}
\newcommand\commentfont[1]{\footnotesize\color{darkgray}\textit{#1}}
\SetCommentSty{commentfont}

In this section, we describe our algorithms for the weighted load balancing
problem.  We start by showing that with the simple reduction in~\cite{HKPR18}
from unweighted to weighted load balancing, our unweighted algorithm
(\Cref{alg:congest_unweighted}) also implies an
$O(\log{n})$-approximate $\polylog(n)$-round \congest{} algorithm for weighted
instances.  We then turn to the main result of this section: an
$O(1)$-approximate $\polylog(n)$-round \local{} algorithm for the weighted load
balancing problem.  We conclude this section with an $O(1)$-approximate
sequential algorithm that runs in near-linear time.

As our goals in this section are to obtain, at best, an $O(1)$-approximation,
we may assume that all client weights are powers of two.  If not, rounding
weights up to the nearest power of two will at most double the approximation
ratio.  We can assume similarly that the maximum weight $W \le n$.  Indeed,
clients with load less than $W / n$ can collectively distribute at most $W$
weight across the servers and can therefore be assigned arbitrarily.  Thus,
when $W > n$, clients can simply rescale their own weight by $n / W$ (and round
it up to the nearest integer).

Throughout this section, we denote by $C_i$ the set of clients whose weight is
exactly $2^i$.  (By our previous assumption, the sets $\{C_i\}$ partition $C$.)
We let $G_i := G[C_i, N(C_i)]$ be the induced graph on $C_i$ and its
neighborhood.

\subparagraph{An $O(\log{n})$-approximation in the \congest{} model.}
We begin with an easy corollary of our unweighted algorithm following a simple
reduction in~\cite{HKPR18}.
\begin{theorem}\label{thm:weighted-reduction}
  In the \congest{} model, an $O(\log{n})$-approximation to the all-norm
  weighted load balancing problem can be computed with high probability in
  $O(\log^5\!{n})$ rounds.
\end{theorem}
\begin{proof}
  Consider the following algorithm: For each weight class $i$, compute an
  assignment $A_i$ of $G_i$ using \Cref{alg:congest_unweighted} by
  treating all clients as having weight 1. Then, have each client in $C_i$
  assign itself according to $A_i$.

  Since all $A_i$'s can be computed in parallel (as the graphs $G_i$ are
  edge-disjoint, only one of the parallel copies need to communicate over an
  edge), the algorithm  runs in $O(\log^5\!{n})$ rounds. We now show that the
  resulting assignment $A$ is $O(\log{n})$-approximate for all norms.

  Fix any $p \ge 1$ including $p = \infty$; let $A^*$ be an assignment for $G$
  with minimum $\ell_p$-norm, and let $A^*_i$ be an assignment for $G_i$ with
  minimum $\ell_p$-norm.  Clearly $\norm{A^*_i}_p \le \norm{A^*}_p$ for all
  $i$.  By \Cref{thm:unweighted-pnorm}, there is a constant $K$ such that
  $\norm{A_i}_p \le K \norm{A^*_i}_p \le K \norm{A^*}_p$.  It follows that \[
    \norm{A}_p = \norm{A_1 + \cdots + A_{\log{n}}}_p \le \norm{A_1}_p + \cdots + \norm{A_{\log{n}}}_p \le K \log(n) \norm{A^*}_p.\qedhere
  \]
\end{proof}

\subparagraph{Preliminaries for the weighted algorithms.}
Though they use entirely different techniques, the \local{} and sequential
algorithms of the next two subsections both follow the same high-level
approach: first compute a split assignment, then round it into an integral one.

\begin{definition}
  Let $G = (C \cup S, E)$ be a bipartite graph with client weights $w : C \to
  \mathbb{Z}^+$. A \emph{split assignment} $\afrac$ in $G$ is a
  client-perfect $(w,\infty)$-matching (so servers have unbounded capacity).
  For every server $s$, the \emph{load} $L_{\afrac}(s)$ is the sum of
  edge-multiplicities incident to $s$.
\end{definition}

Notice that split assignments are a relaxation of standard assignments by
allowing clients to be assigned to several different servers at once,
contributing an integral load to each server, provided that the total load
distributed by the client does not exceed its weight.

We will also need the following notion.  Define the \emph{client-expanded
graph} $\widetilde{G}$ of $G$ as the graph formed by making $w(c)$ copies of
each client $c$.  Formally, for each $c \in C$, the client-expanded graph has
vertices $c_1, \dots, c_{w(c)}$ and an edge between $c_i$ and $s$ for all $i
\in [w(c)]$ if and only if $G$ has an edge between $c$ and $s$. 

\begin{observation}\label{obs:frac}
  A split assignment $\afrac$ in $G$ corresponds to an integral assignment
  in the client-expanded graph $\tilde{G}$ with the same server loads. Thus,
  since $\tilde{G}$ is unweighted, by \Cref{lem:optimal-exists} there
  exists an all-norm optimal split assignment $\afracopt$.
\end{observation}

\subsection{An $O(1)$-approximation in the \local{} model}

Our main result in this section is the following theorem.

\begin{theorem}\label{thm:pnorm-weighted}
  In the \local{} model, there is a randomized algorithm
  (\Cref{alg:local_weighted}) that with high probability computes an
  $O(1)$-approximation to the weighted all-norm load balancing problem in
  $O(\log^3\!{n})$ rounds. 
\end{theorem}

We will need the next rounding lemma to describe our algorithm; the proof is
standard.

\begin{lemma}\label{lem:cyclecancelling}
  If $G = (C \cup S, E)$ contains a client-perfect $(\kappa, \tau)$-matching
  $x$, then there exists an assignment $A : C \to S$ such that for all servers
  $s \in S$,
  \begin{align*}
    L_A(s) \le \tau(s) + \max_{c \in A^{-1}(s)} \kappa(c).
	\end{align*}
\end{lemma}

\begin{proof}
Consider the set of edges $F$ in the support of $x$.  If $C \subseteq F$ is a
cycle, we can alternately increase and decrease the value of $x(e)$ on each
edge $e$ of the cycle by $\min_{f\in C} x(f)$ to break the cycle without
changing $x(\delta(v))$ for any $v \in V$ (this cycle can only be of even
length as the input graph is bipartite).  Thus, we may assume that the support
of $x$ has no cycles and thus is a forest. 

We can next turn $F$ into a collection of stars centered on servers. This done
by rooting each tree $T$ in the support of $F$ arbitrarily, picking each server
$s$ which has a client parent-node $c$, and setting the edge $x(cs) =
\kappa(c)$ and $x(cs') = 0$ for all other $s' \in N(c)$. This clearly satisfies
the requirement of client $c$ and the load on server $s$ can only ever be
increased by $\max_{c \in A^{-1}(s)} \kappa(c)$ as each server can only have
one parent client. At this point, in $F$, any client is assigned to exactly one
server and thus we obtain an integral solution in which the load of any server
$s$ is at most $\tau(s) + \max_{c \in A^{-1}(s)} \kappa(c)$, finalizing the
proof. 
\end{proof}

Our \local{} algorithm consists of two main parts, an algorithm for solving the
split load balancing problem and a rounding procedure, which we describe
now in turn.

\subparagraph{Computing a split assignment.} 

The first step of the \local{} algorithm is to compute an assignment
$\tilde{A}$ in the client-expanded graph $\widetilde{G}$ of $G$ using
\Cref{alg:congest_unweighted}.  Note that in the \local{} model, each
client $c$ can simulate all ``new'' clients $c_1,\ldots,c_{w(c)}$ in
\Cref{alg:congest_unweighted} without any overhead in the round
complexity.\footnote{We remark that computing this assignment is the only step
of our weighted algorithm that does not run efficiently in the \congest{}
model, precisely because this simulation not possible in the \congest{}
model in $\polylog(n)$ rounds.} As mentioned in \Cref{obs:frac}, the
assignment $\tilde{A}$ corresponds to a split assignment with the same
server loads.  To limit the amount of notation in the algorithm description, we
will sometimes refer to $\tilde{A}$ as a split assignment in $G$, although
formally it is an assignment in $\tilde{G}$. 

The guarantees of \Cref{alg:congest_unweighted} tell us that
$\tilde{A}$ has small $\ell_p$-norm.  The next step is to use $\tilde{A}$ to
find an integral assignment without much loss in the norm.

\subparagraph{A ``rounding'' procedure.} 
We would now ideally round the split assignment $\tilde{A}$ into an
integral assignment, but even in the \local{} model we cannot afford to run
such a procedure directly.  The fact that a good rounding exists, however, is
enough for us to apply \Cref{lem:expansion} to obtain a similarly good
assignment, as we show below.

For each $i$, let $\tilde{A}_i$ be $\tilde{A}$ restricted to $G_i$.
\Cref{lem:cyclecancelling} states that there is a way to round
$\tilde{A}_i$ into an assignment with load $\tau_i(s) = L_{\tilde{A}_i}^i(s) +
2^i$ for servers $s$ assigned to by $\tilde{A}_i$ and $\tau_i(s) = 0$ for the
remaining servers.  Treating the clients as unweighted, $\tilde{A}_i$
corresponds to a $(1, \lceil 2^{-i} \tau_i\rceil)$-matching.  We now compute a
$(1, 2 \lceil 2^{-i}\tau_i \rceil)$-matching $x_i$ with no augmenting paths of
length $4\lceil \log{n} \rceil + 1$ or smaller.  By \Cref{lem:expansion},
each $x_i$ is client-perfect, inducing an (integral) assignment $A_i$ in $G_i$.
Lastly, each client in $C_i$ assigns itself in accordance with $A_i$ to produce
the global assignment $A$.  See \Cref{alg:local_weighted}.

\begin{figure}[t]
  \begin{algorithm}[H]
    \caption{\label{alg:local_weighted} Approximate weighted (all-norm) load balancing in the \local{} model.}
    emulate \Cref{alg:congest_unweighted} on $\widetilde{G}$ to compute an assignment $\tilde{A}$\\
    \For{$i \in \{1, 2, 4, \dots, n\}$ in parallel}{
      let $\tilde{A}_i$ be $\tilde{A}$ restricted to $\widetilde{G}_i$\\
      let $\tau_i(s) = \begin{cases}
        L_{\tilde{A}_i}(s) + 2^i, &\text{if $L_{\tilde{A}_i}(s) > 0$}\\
        0, &\text{otherwise}\\
      \end{cases}$\nllabel{lin:tau}\\
      \Comment{by \Cref{lem:cyclecancelling}, an assignment with load vector point-wise less than $\tau_i$ exists in $G_i$}
      \Comment{therefore, scaling clients in $G_i$ to weight 1, a $(1, \lceil 2^{-i}\tau_i\rceil)$-matching exists}
      treating $G_i$ as unweighted, compute a $(1, 2 \lceil 2^{-i}\tau_i\rceil)$-matching $x_i$ in $G_i$ with no augmenting paths of length $4\lceil \log{n} \rceil + 1$\nllabel{lin:compute-matching}\\
      \Comment{by \Cref{lem:expansion}, $x_i$ is client-perfect}
      let $A_i$ be the assignment induced by $x_i$\nllabel{lin:assignment}\\
      assign each $c \in C_i$ to $A_i(c)$
    }
  \end{algorithm}
\end{figure}

To formalize the logic of the algorithm, we make a few claims that together
will imply the algorithm's correctness. The first claim ensures that the
algorithm produces a proper assignment.

\begin{claim}
  \Cref{alg:local_weighted} assigns every client to some server.
\end{claim}
\begin{proof}
  We need to show that the matching $x_i$ computed in
  \Cref{lin:compute-matching} of \Cref{alg:local_weighted} is
  client-perfect.  Consider $\tau_i$ from \Cref{lin:tau} of
  \Cref{alg:local_weighted}.  Viewing $\tilde{A}_i$ as a
  client-perfect $(w, L_{\tilde{A}_i})$-matching,
  \Cref{lem:cyclecancelling} guarantees that there is an assignment
  wherein each server $s$ has load at most $\tau_i(s)$.
  
  Because all clients in $G_i$ have the same weight, we can interpret the
  assignment from \Cref{lem:cyclecancelling} as a client-perfect $(1,
  \lceil 2^{-i} \tau_i\rceil)$-matching in the unweighted graph $G_i$.  When
  treating clients as unweighted, server capacities are always bounded by $n$,
  and so by \Cref{lem:expansion}, if $x_i$ has no augmenting paths of
  length $\le 4\lceil \log{n}\rceil + 1$, it follows that $x_i$ is
  client-perfect.
\end{proof}

The next claim shows that the assignment produced is $O(1)$-approximate.

\begin{claim}
  There is a universal constant $C$ such that for all $p \ge 1$, including $p =
  \infty$, the assignment $A$ produced by \Cref{alg:local_weighted}
  satisfies $\norm{A}_p \le C \norm{A^*}_p$, where $A^*$ is an
  $\ell_p$-norm-minimizing assignment.
\end{claim}

\begin{proof}
  Fix $p \ge 1$ (including $p=\infty$).  Let $A^*$ and $\tilde{A}^*$ be
  assignments for $G$ and $\tilde{G}$, respectively, that minimize the
  $\ell_p$-norm.  For brevity, we omit the subscript $p$ when writing norms
  with the understanding that all norms in this proof are $\ell_p$-norms.
  We will also treat the client weight function $w$ as a vector so that we can
  write its norm as $\norm{w}$.

  Our strategy is to decompose the final assignment $A$ into two parts and
  bound the norms of those parts separately.  First, we decompose each
  assignment $A_i$ of \Cref{lin:assignment}.  We define the first part,
  $\rho_i$, by $\rho_i(s) = 2^i$ if $s$ is assigned to by $A_i$ and $\rho_i(s)
  = 0$ otherwise.  In other words, $\rho_i$ has the same support as the load
  vector $L_{A_i}$ of $A_i$, but all of its nonzero entries are $2^i$.  The
  second part, $\mu_i$, docks $2^{i+1}$ from the support of $L_{A_i}$: $\mu_i =
  L_{A_i} - 2\rho_i$.  Letting $\mu = \sum_i \mu_i$ and $\rho = \sum_i \rho_i$,
  we have that $\norm{A} = \norm{\mu + 2\rho} \le \norm{\mu} + 2\norm{\rho}$.
  It therefore suffices to show that $\norm{\mu}$ and $\norm{\rho}$ both
  $O(1)$-approximate $\norm{A^*}$.

  Let us first bound $\norm{\mu}$.  For any server $s$ assigned to by $A_i$, we
  have \begin{align*}
    \mu_i(s) &= L_{A_i}(s) - 2^{i+1}\\
      &\le 2^{i+1}\lceil 2^{-i} \tau_i(s)\rceil - 2^{i+1}\\
      &\le 2^{i+1}\lceil 2^{-i} L_{\tilde{A}_i}(s) + 1\rceil - 2^{i+1}\\
      &\le 2^{i+1} 2^{-i+1} L_{\tilde{A}_i}(s) + 2^{i+1} - 2^{i+1} \tag{$\lceil x\rceil \le 2x$ for all $x \ge 1$} \\
      &\le 4 L_{\tilde{A}_i}(s).
  \end{align*}
  For any server $s$ \textit{not} assigned to by $A_i$ we have $\mu_i(s) = 0$,
  and so trivially $\mu_i(s) \le 4 L_{\tilde{A}_i(s)}$ for such $s$.
  Therefore, $\mu(s) = \sum_i \mu_i(s) \le \sum_i 4L_{\tilde{A}_i}(s) =
  4L_{\tilde{A}}(s)$.  Using \Cref{thm:hkpr}, it follows that
  $\norm{\mu} \le 4 \norm{\tilde{A}} \le 32 \norm{\tilde{A}^*} \le
  32\norm{A^*}$.

  We now bound $\norm{\rho}$.  Define $\rho^*(s) = \max_{c \in A^{-1}(s)}
  w(c)$.  Note that $\rho^*$ is the load vector of a ``partial'' assignment
  (not all clients are assigned) that assigns to each server at most once.
  Since $w$ can be interpreted as the load vector of an assignment that assigns
  \textit{every} client to a unique server, we have $\norm{\rho^*} \le
  \norm{w}$.  Now observe that $\rho(s) = \sum_{i=1}^{\log\rho^*(s)} \rho_i(s)
  \le 2\rho^*(s)$ simply because $\rho_i(s)$ is either 0 or $2^i$ for each $i$.
  To complete the bound, notice that $\norm{w} \le \norm{A^*}$; the best
  (hypothetical) assignment would assign every client to a unique server,
  resulting in value $\norm{w}$.  Putting things together, we have shown that
  $\norm{\rho} \le 2\norm{A^*}$.
\end{proof}

It remains to bound the round-complexity of the algorithm.

\begin{claim}
  \Cref{alg:local_weighted} takes $O(\log^3\!{n})$ rounds in the
  \local model.
\end{claim}
\begin{proof}
  In the \local model, we can easily emulate
  \Cref{alg:congest_unweighted} (or any algorithm) on the
  client-expansion $\tilde{G}$ at no extra cost; any communication across an
  edge $cs$ simply needs to specify which $c_i$ in the expansion the message is
  to/from.  Since $W \le n$, \Cref{alg:congest_unweighted} still runs
  in $O(\log^3\!{n})$ rounds in the \local model (see
  \Cref{rem:alg1-local}).  The main for-loop is run in parallel, and so
  we only need to bound the round-complexity of its body.
  \Cref{lin:compute-matching} is the only line inside the loop that
  requires (additional) communication, and this again only takes
  $O(\log^3\!{n})$ rounds.  The total round-complexity is therefore
  $O(\log^3\!{n})$.
\end{proof}

This concludes the proof of \Cref{thm:pnorm-weighted}. 

\subsection{An $O(1)$-approximate $O(m\log^3\!{n})$-time sequential algorithm}

We now show that our approach can also be used to compute an
$O(1)$-approximation to the  weighted all-norm load balancing problem in
near-linear time in the standard sequential setting, proving the following
theorem. 

\begin{theorem}\label{thm:sequential}
  In the standard sequential model, there is a deterministic algorithm to
  compute an $O(1)$-approximate solution to the weighted all-norm load
  balancing problem that runs in $O(m\log^3\!{n})$ time.
\end{theorem}

Previously, Azar et al.~\cite{AERW04} showed a 2-approximate algorithm for this
problem, which runs in two phases: (1) compute an optimal fractional assignment
and (2) round the fractional assignment, which incurs a 2-approximation. But
their algorithm computes the optimal fractional assignment using the ellipsoid
method to solve a linear program with exponentially many constraints, and hence
incurs a large polynomial runtime.

Our algorithm uses the same rounding procedure as~\cite{AERW04}, but instead of
computing an exact fractional assignment, we compute an $O(1)$-approximate
split assignment in near-linear time by simulating our distributed approach in
the sequential setting. To this end, we will need the following subroutine:

\begin{lemma}\label{lem:sequential-flow}
  Given any bipartite graph $G = (C \cup S, E)$ and capacity functions
  $\kappa$, $\tau$, it is possible to compute a $(\kappa,\tau)$-matching with
  no augmenting paths of length $\leq 9\log(n)$ in $O(m\log^2\!n)$ time in the
  sequential setting.
\end{lemma}

\begin{proof}
Note that a $(\kappa,\tau)$-matching corresponds to the following flow problem.
Every edge in $E$ gets infinite capacity; there is a dummy source $v_s$ and for
every client $c \in C$ there is an edge $(v_s,c)$ of capacity $\kappa(c)$;
there is also a dummy sink $v_t$ and for every server $s \in S$ there is an
edge from $s$ to $v_t$ of capacity $\tau(s)$. It is immediate to verify that
any $v_s$-$v_t$ flow in this network corresponds to a $(\kappa,\tau)$-matching
and vice versa. 

We now show how to compute a solution to this flow problem that contains no
augmenting paths of length $9\log(n) \geq 8\log(n) + 2$ which corresponds to
the desired $(\kappa, \tau)$-matching.

The algorithm simply runs $9\log(n)$ successive iterations of blocking flow. A
blocking flow in a capacitated graph can be computed in $O(m\log{n})$ time
using the dynamic tree structure of Sleator and Tarjan~\cite{SleatorT83}.
\end{proof}

We are now ready to show our algorithm to compute a split assignment.


\begin{lemma}\label{lem:sequential-fractional}
  Let $G = (C \cup S, E)$ be a bipartite graph with client-weights $w(C)$.
  There exists a sequential algorithm that, for some constant $K$, in
  $O(m\log^3\!{n})$ time computes a split assignment $\afrac$ such that
  $\norm{L_{\afrac}}_p \leq K \norm{L_{\afracopt}}_p$ for every $p \geq 1$
  (including $p = \infty$).
\end{lemma}

\begin{proof}
Recall from Observation~\ref{obs:frac} that the optimal split assignment
$\afracopt$ corresponds to an optimal (integral) assignment $\tilde{A^*}$ in
the client-expanded graph $\tilde{G}$; server loads in the two solutions are
the same, so $\norm{L_{\afracopt}}_p = \norm{L_{\tilde{A^*}}}_p$. We obtain our
split assignment $\afrac$ by simulating
Algorithm~\ref{alg:congest_unweighted} on the graph $\tilde{G}$: by
Theorem~\ref{thm:unweighted-pnorm}, this yields the desired $O(1)$-approximation.
We now describe how to execute the simulation in the sequential model and how
to convert between the perspectives of split assignment in $G$ and
integral assignment in $\tilde{G}$. 

Firstly, in Line 2 of Algorithm~\ref{alg:congest_unweighted}, we need to a
compute a $(1,B)$-matching in $\tilde{G}$ with no short augmenting paths. This
is equivalent to a $(w,B)$-matching in $G$, which we compute in $O(m\log^2\!n)$
time using Lemma~\ref{lem:sequential-flow}. 

Secondly, in Line 4 of Algorithm~\ref{alg:congest_unweighted}, each client-copy
$\tilde{c}$ in $\tilde{G}$ must find the minimum $B$ such that $\tilde{c}$ is
matched in $x_B$. We need to convert this line to the language of split
assignments. In particular, note that in our sequential simulation, $x_B$ is a
$(w,B)$-matching in $G$ rather than a $(1,B)$-matching in $\tilde{G}$. It is
easy to see that the following simulates Line 4. For each client $c$ in $G$,
let $s_B(c)$ be the set of servers incident to $c$ in $x_B$: if an edge $cs$
has multiplicity $\alpha$ in $x_B$, then $s$ appears $\alpha$ times in
$s_B(c)$. To construct the split assignment $\afrac$, first assign $c$ to
the server in $s_1(c)$ (if any). Then assign $c$ to an arbitrary $|s_2(c)| -
|s_1(c)|$ servers from $s_2(c)$, an arbitrary $|s_4(c)| - |s_2(c)|$ servers
from $s_4(c)$, and more generally an arbitrary $|s_B(c)| - |s_{B/2}(c)|$
servers from $s_B(c)$. It is not hard to check that the resulting split
assignment is equivalent to some integral assignment in $\tilde{G}$ formed by
executing Line 4 of Algorithm~\ref{alg:congest_unweighted} in $\tilde{G}$. It
is also easy to see that for each $x_B$ the assignments can be performed in
$O(m)$ time, for a total of $O(m\log{n})$ time.

The running time of the algorithm is thus dominated by the time for computing
matchings $x_B$. Each takes $O(m\log^2\!{n})$ time to compute
(Lemma~\ref{lem:sequential-flow}), and there are $O(\log(nW)) = O(\log{n})$
values of $B$, so the total run-time is $O(m\log^3\!{n})$. \qedhere


\end{proof}

Finally, we round the split assignment to an integral assignment using the
rounding procedure of~\cite{AERW04}, which is described in the proof of
Lemma~\ref{lem:cyclecancelling}. The rounding procedure has two steps: cycle
cancelling and computing a matching in a tree. The second can clearly be done
in $O(m)$ sequential time. Cycle cancelling can be done deterministically in
$O(m\log{n})$ time (see, e.g.,~\cite{KangP15}). The total time for rounding is
thus $O(m\log{n})$.  (Note that in the distributed setting we only relied on
the \emph{existence} of such a rounding procedure, because it is unclear how to
implement cycle canceling efficiently in the \local{} model.)

Following the exact same argument as in~\cite{AERW04} or in the proof of
Lemma~\ref{thm:pnorm-weighted} of this paper, since our split assignment was an
$O(1)$-approximation (Theorem~\ref{lem:sequential-fractional}), the integral
assignment formed by rounding also yields a $O(1)$-approximation. This
concludes the proof of Theorem~\ref{thm:sequential}.

\section{Acknowledgements}

We thank Shyamal Patel and Cliff Stein for pointing out an error in the proof
of Lemma~\ref{lem:sat-opt} in the conference version of the paper.

\bibliographystyle{abbrv}
\bibliography{main}

\clearpage
\appendix

\part*{Appendix}
\section{Distributed Backup Placement with Replication Factor}\label{app:dbp}

In the \emph{distributed backup placement problem}, the input is a graph $G =
(V, E)$.  There are a set of nodes $C \subseteq V$ called the \emph{clients}
which host files that should be backed up on over a set of nodes $S \subseteq
V$ called the \emph{servers}.  Unlike the load balancing problem, the clients
and servers here need not partition $V$ or even be disjoint.  Each client $c$
hosts a file of size $w(c)$ and may only backup the file on adjacent servers;
we will refer to $w(c)$ as the \emph{weight} of client $c$.  When all of the
client-weights are the same, we call the instance \emph{uniform}.  Finally, a
\emph{replication factor} $r$ specifies the number of \textit{distinct} servers
that each client must be backed up on; each client must be assigned to $r$
distinct adjacent servers.  The goal is to minimize the maximum server load in
the resulting assignment; as before, the load of a server is the sum of the
client-weights assigned to it. Following the paper of Halld{\'o}rsson et
al.~\cite{HKPR18}, we assume that every client has degree at least $r$, since
otherwise there is no solution to the problem.

As was shown by Halld{\'o}rsson et al.~\cite{HKPR18}, when the replication
factor is one, the distributed backup placement problem can be reduced to the
load balancing problem in a natural way. Form a bipartite graph $G'$ of the
clients and servers (a node for a server may appear on both sides of the
partition).  Clients and servers are adjacent in the new graph if and only if
they were adjacent in $G$.  A solution to the load balancing problem in $G'$
directly corresponds to a solution to the distributed backup placement problem
in $G$.  The results in our paper therefore immediately give improved bounds
for the distributed backup placement problem with replication factor one.

We show that our approach can also be used to handle replication factor larger
than 1. For simplicity, we only extend our results that are most close related
to the existing state of the art by Halld{\'o}rsson et al.~\cite{HKPR18}. Their
paper shows that for any replication factor, in $\polylog(n)$ rounds, it is
possible to compute a $O(\log{n}/\log\log{n})$-approximation to distributed
backup placement with uniform client weights. They also show a simple reduction
which gives a $O(\log^2\!{n}/\log\log{n})$-approximation for general client
weights. Existing results on distributed backup placement focus only on
minimizing maximum load (not general $\ell_p$-norm), so we will do the same.

In this section, we show that our \Cref{thm:hkpr} can be extended to the
problem of distributed backup placement with arbitrary replication factor, thus
improving upon the result of Halld{\'o}rsson et al.~\cite{HKPR18}. In
particular, we achieve the following:

\begin{theorem}\label{thm:dbp-unweighted}
  Given an instance of the backup placement problem with uniform client weights
  and replication factor $r$, an $8$-approximate solution can be found w.h.p.\
  in the \congest{} model within $O(\log^5\!{n})$ rounds.
\end{theorem}

Using the same reduction from weighted clients to unweighted clients as in
\Cref{thm:weighted-reduction}, we also obtain the following improvement
over Corollary 27 in~\cite{HKPR18}.

\begin{theorem}\label{thm:dbp-weighted}
  Given an instance of the backup placement problem with non-uniform client
  weights and replication factor $r$,  an $O(\log{n})$-approximate solution can
  be found w.h.p.\ in the \congest{} model within $O(\log^5\!{n})$ rounds.
\end{theorem}

\subsection{The Setup}

Our Algorithm for \Cref{thm:dbp-weighted} follows the same structure as
\Cref{sec:unweighted} for load balancing. We just need small
modifications to handle arbitrary replication factor $r$, rather than the
replication factor $1$ of standard load balancing. 

To this end, we first generalize our notion of a $(1,B)$ matching

\begin{definition}\label{def:dbp-matching}
  Given any positive integers $B,r$, we say that $x \subseteq E$ is a
  $(1,B,r)$-matching if every client has degree at most $r$ in $x$, and every
  server has degree at most $B$. Note that in this definition, every edge has
  multiplicity at most 1, which is why a $(1,B,r)$-matching is different from a
  $(r,B)$-matching. We say that $x$ is client-perfect if every client has
  degree $r$ in $x$, and we say that a client is unsaturated if it has degree
  strictly less than $r$. An $x$-augmenting path is defined the same way as
  before. 
\end{definition}

\begin{observation}
  A client-perfect $(1,B,r)$-matching is a solution to backup placement with
  replication $r$ that has maximum server load $B$. 
\end{observation}

Our structural lemma for $(\kappa,\tau)$-matchings (\Cref{lem:expansion})
can easily be extended to the case of $(1,B,r)$-matchings; the proof is the
same.

\begin{lemma}[Extension of \Cref{lem:expansion}]\label{lem:dbp-expansion}
  If $G$ contains a client-perfect $(1, B, r)$-matching and $x$ is a
  $(1,2B,r)$-matching, then either $x$ is client-perfect or there is an
  $x$-augmenting path of length at most $4\lceil \log{n} \rceil + 1$.
\end{lemma}

Finally, given any positive integers $B,r$, we can generalize
\Cref{lem:hksr} to efficiently compute a $(1,B,r)$-matching with no short
augmenting paths: 

\begin{lemma}[\cite{HKPR18}] (Extension of \Cref{lem:hksr})\label{lem:dbp-hksr}
  There exists an $O(k^3\!\log{n})$-round randomized algorithm in the
  \congest{} model that, with high probability, given a graph $G = (C \cup S,
  E)$, and positive integers $B,r,k$, computes $(1, B,r)$-matching with no
  augmenting paths of length less than $k$.
\end{lemma}

The proof of the above lemma is the same as that of \Cref{lem:hksr}. In
particular, the proof combines two existing results. The algorithm of Lotker et
al.~\cite{LPP15} computes a (1,1)-matching with no augmenting paths of length
$\leq k$ in $O(k^3\log{n})$ rounds. The paper of Halld{\'o}rsson et
al.~\cite{HKPR18} then shows a black-box extension from $(1,1)$-matching to
$(1,B,r)$-matching which does not increase the round-complexity. In particular,
they reduce from a problem called $f$-matching, which captures the setting
where there can be multiple copies of \emph{both} clients and servers, but
where each edge can only be used once; by using $r$ copies of each client and
$B$ copies of each server, we get a $(1,B,r)$-matching. See~\cite{HKPR18} for
more details.

\subsection{The Algorithm}

Our algorithm is basically the same as our algorithm for load balancing
(\Cref{alg:congest_unweighted}). We now describe the changes we need
to make to handle replication factor $r$.

Firstly, in \Cref{line:unweighted-augment} of
\Cref{alg:congest_unweighted}, instead of computing a $(1,B)$-matching
$x_B$ with no short augmenting paths, we invoke  \Cref{lem:dbp-hksr} to
compute a $(1,B,r)$-matching $x_B$ with no short augmenting paths.

Secondly, in \Cref{line:unweighted-assign} of
\Cref{alg:congest_unweighted}, each client $c$ locally finds the
minimum $B$ such that $c$ is matched $r$ times in $x_B$ and then assigns itself
to those $r$ servers. (Note that if $c$ is assigned to $<r$ servers in some
$x_{B'}$, then $c$ simply ignores $x_{B'}$.)

The round-complexity of the algorithm is clearly the same $O(\log^5\!{n})$ as
in \Cref{alg:congest_unweighted}. The approximation analysis is
exactly the same as in the proof of \Cref{thm:hkpr}, so the algorithm
computes an $8$-approximate assignment for backup placement.

\end{document}